\documentclass[prx,twocolumn,showpacs,amsmath,amssymb,superscriptaddress,floatfix,longbibliography,nofootinbib]{revtex4-2}
\usepackage{times}
\usepackage{helvet}
\usepackage{courier}
\usepackage{hyperref}
\hypersetup{
  colorlinks=true,    % false: boxed links; true: colored links
  linkcolor=cyan,     % color of internal links
  citecolor=magenta,    % color of links to bibliography
  filecolor=magenta,   % color of file links
  urlcolor=cyan,      % color of external links
  runcolor=cyan
}
\usepackage{graphicx}
\usepackage{orcidlink}
\usepackage{algorithm}
\usepackage{algorithmic}
\usepackage{newfloat}
\usepackage{listings}
\usepackage{amsmath}
\usepackage{amsfonts}
\usepackage{amssymb}
\usepackage{amsthm}
\usepackage{braket}
\usepackage[capitalize]{cleveref}
\usepackage{xcolor}
\usepackage{booktabs}
\usepackage{tikz}
\usepackage{mathtools}
\usepackage{subfigure}

\newtheorem{mytheorem}{Theorem}

\newtheorem{mycorollary}{Corollary}

\newtheorem{mydefinition}{Definition}

\newcommand{\bes} {\begin{subequations}}
\newcommand{\ees} {\end{subequations}}

\def\>{\rangle}
\def\<{\langle}

\newcommand{\ketb}[2]{|{#1}\>\!\<#2|}

\begin{document}

\title{Efficient Chromatic-Number-Based Multi-Qubit Decoherence and Crosstalk Suppression}

\author{Amy F. Brown\,\orcidlink{0000-0001-6664-6494}}
%\email{afbrown@usc.edu}
\affiliation{Department of Physics \& Astronomy, University of Southern California,
Los Angeles, California 90089, USA}
\affiliation{Center for Quantum Information Science \& Technology, University of
Southern California, Los Angeles, California 90089, USA}
\author{Daniel A. Lidar\,\orcidlink{0000-0002-1671-1515}}
\affiliation{Department of Physics \& Astronomy, University of Southern California,
Los Angeles, California 90089, USA}
\affiliation{Center for Quantum Information Science \& Technology, University of
Southern California, Los Angeles, California 90089, USA}
\affiliation{Department of Electrical \& Computer Engineering, University of Southern California,
Los Angeles, California 90089, USA}
\affiliation{Department of Chemistry, University of Southern California, Los Angeles,
California 90089, USA}

\begin{abstract}
The performance of quantum computers is hindered by decoherence and crosstalk, which cause errors and limit the ability to perform long computations. Dynamical decoupling is a technique that alleviates these issues by applying carefully timed pulses to individual qubits, effectively suppressing unwanted interactions. However, as quantum devices grow in size, it becomes increasingly important to minimize the time required to implement dynamical decoupling across the entire system. Here, we present ``Chromatic-Hadamard Dynamical Decoupling'' (CHaDD), an approach that efficiently schedules dynamical decoupling pulses for quantum devices with arbitrary qubit connectivity. By leveraging Hadamard matrices, CHaDD achieves a circuit depth that scales linearly with the chromatic number of the connectivity graph for general two-qubit interactions, assuming instantaneous pulses. This includes $ZZ$ crosstalk, which is prevalent in superconducting QPUs. CHaDD's scaling represents an exponential improvement over all previous multi-qubit decoupling schemes for devices with connectivity graphs whose chromatic number grows at most polylogarithmically with the number of qubits. For graphs with a constant chromatic number, CHaDD's scaling is independent of the number of qubits. We report on experiments we have conducted using IBM QPUs that confirm the advantage conferred by CHaDD.
Our results suggest that CHaDD can become a useful tool for enhancing the performance and scalability of quantum computers by efficiently suppressing decoherence and crosstalk across large qubit arrays.
\end{abstract}

\maketitle

\section{Introduction}
Quantum computers contain interconnected qubits that experience decoherence and control errors due to spurious interactions with the environment and surrounding qubits.
This adversely affects such devices' ability to retain or process information for periods of time longer than the decoherence timescale, which is necessary for achieving a quantum advantage over classical computers~\cite{Preskill:2018aa}.
In order to surmount this obstacle, considerable attention has been devoted to developing methods that mitigate these deleterious effects, among which dynamical decoupling (DD) has recently been shown to play an important role. DD is an error suppression technique that averages out undesired Hamiltonian terms by applying deterministic~\cite{Viola:98, Viola:99, Zanardi:1999fk, Vitali:99} or random sequences of scheduled pulses \cite{Viola:2005:060502,Santos:2006:150501}. Originally designed to suppress low-frequency noise in nuclear magnetic resonance \cite{Hahn:50, Carr:54, CPMG1958,Maudsley:1986ty}, DD effectively decouples the quantum system from both external noise due to decoherence \cite{Suter:2016aa, Pokharel2018, DD-survey,raviVAQEMVariationalApproach2021, tong2024empirical} and internal noise due to undesired always-on interactions~\cite{Jones_1999, Tsunoda_2020, Tripathi2022, Zhou2023, niu2024multiqubit, evert2024syncopated}.
DD has recently found numerous applications in quantum information processing, e.g., in noise characterization \cite{Byl11a, Alvarez:2011aa, Paz-Silva2014, Norris2016, TripathiTransmon2024, Gross2024} and improving algorithmic performance
\cite{jurcevicDemonstrationQuantumVolume2021, pokharel2022demonstration, Pokharel:2024aa, singkanipa2024demonstration, baumer2023efficient, baumer2024quantum}.

Most of the attention in improving performance through DD has focused on the development of DD sequences that achieve high perturbation-theory-order noise cancellation~\cite{Khodjasteh:2005xu, Khodjasteh:2007zr, Uhrig:2007qf, Biercuk:09, UL:10, West:10, Wang:10, West:2010:130501,Xia:2011uq,Kuo:2012rf}, robustness to pulse errors~\cite{Alvarez:2010ve,Quiroz:2013fv,Genov:2017aa,Genov:2019aa}, or reduction of the requirements for quantum error correction~\cite{KhodjastehLidar:03,Ng:2011dn,Paz-Silva:2013tt,Unden:2016aa,Conrad:2021aa} and error avoidance~\cite{Mena-Lopez:2023aa,quiroz2024dynamically,han2024protecting}. Much less attention has been paid to developing efficient DD pulse sequences for simultaneously decoupling multiple interconnected qubits, and studies of this topic date primarily to early results using Hadamard matrices and orthogonal arrays~\cite{Jones_1999,Leung:01,Stollsteimer:01,Rotteler:2006aa,Wocjan:2006aa,qec-chap15,Bookatz:2016aa,PhysRevA.84.062323}, along with more recent developments~\cite{Paz-Silva:2016aa,Tsunoda_2020,Zhou2023, Shirizly:2024aa, niu2024multiqubit,evert2024syncopated}. This problem grows in complexity for qubit layouts that are more highly connected, i.e., as the qubit graph degree grows. It is generally recognized that a higher graph degree is desirable~\cite{Linke:2017aa,boothby2020nextgeneration}, as this reduces the overhead associated with coupling geometrically distant qubits. 

In this work, we introduce Chromatic-Hadamard Dynamical Decoupling (CHaDD), which provides an efficient solution to completely decoupling an arbitrary qubit interaction graph $G=(V,E)$, where $V$ is the vertex (i.e., qubit) set and $E$ the edge (i.e., qubit-qubit coupling) set.
There are two relevant notions of efficiency: (i) the total power consumed by a pulse sequence per decoupling cycle, quantified in terms of the pulse repetition rate (PRR), i.e., the average number of pulses per unit time; and (ii) the circuit depth per cycle, where simultaneous pulses are counted as a single pulse. 

Concerning circuit depth, until recently, efficient schemes were thought to be at most linear in the total number of qubits $n=|V|$~\cite{Wocjan:2006aa,qec-chap15,Bookatz:2016aa} for instantaneous (``bang-bang''~\cite{Viola:98}) pulses.
Then, it was established that the scaling of complete decoupling of a $ZZ$ crosstalk graph~\cite{Sarovar2020detectingcrosstalk} could be stated in terms of the chromatic number of the graph $\chi(G)$ (an idea dating back to Refs.~\cite{Jones_1999,Tsunoda_2020}), i.e., the minimum number of distinct colors required to properly color the graph, instead of the number of qubits, albeit with exponential scaling, $2^{\chi(G)}$~\cite{evert2024syncopated}. 

Here, we show that it is possible to achieve \emph{efficient}, universal, first-order decoupling. Namely, we prove that in general, the circuit depth scales linearly with the chromatic number, with an improved prefactor for single-axis decoupling (e.g., the special case of $ZZ$ crosstalk). 
This is an exponential improvement over Ref.~\cite{evert2024syncopated}'s single-axis crosstalk suppression result, as well as over all previous multi-axis decoupling schemes building on Hadamard matrices or orthogonal arrays~\cite{Leung:01,Stollsteimer:01,Rotteler:2006aa,Wocjan:2006aa,qec-chap15,Bookatz:2016aa} for graphs with a chromatic number that scales at most polylogarithmically in $n$. The schemes of Refs.~\cite{Wocjan:2006aa,qec-chap15,Bookatz:2016aa} achieve parity with CHaDD when $\chi(G) \sim n$, such as trapped-ion or neutral atom devices with all-to-all connectivity~\cite{Linke:2017aa,Bluvstein:2022aa}. In devices with a constant chromatic number, such as all current superconducting quantum computing hardware, CHaDD's scaling is independent of $n$, providing a significant advantage over decoupling schemes that scale linearly with the number of qubits. In such cases, CHaDD can further help diagnose and suppress crosstalk between non-natively coupled qubits corresponding to an effective connectivity graph with a higher chromatic number. 

Going beyond CHaDD's efficient circuit depth scaling, we also consider the PRR (i.e., power consumption) metric at fixed chromatic number. We demonstrate both theoretically and experimentally that CHaDD achieves similar or better performance at a lower PRR than standard, ``achromatic'' DD pulse sequences, i.e., sequences that are blind to the chromatic number of the underlying qubit connectivity graph. Namely, we show that CHaDD reduces the PRR by $20-33\%$, depending on the specific CHaDD variant.

The structure of this paper is as follows. In \cref{sec:background}, we provide pertinent background on graph coloring, the error model, and dynamical decoupling. In \cref{sec:results}, we define CHaDD and several variants designed to suppress multi-axis decoherence and to improve robustness against pulses errors. We prove our main theorems concerning the circuit depth efficiency of CHaDD and discuss several illustrative examples. In \cref{sec:experiments}, we provide experimental tests of CHaDD and its variants using IBM superconducting quantum processing units (QPUs) and demonstrate that, as predicted by our theory, CHaDD exhibits a better PRR than previously known sequences. We conclude in \cref{sec:conc} and discuss various possible generalizations and open problems suggested by our work. The Appendix provides additional technical details in support of the main text as well as additional confirmation of our experimental results obtained on different IBM QPUs.

\section{Background}
\label{sec:background}

\subsection{Graph coloring}
Coloring is the assignment of colors (labels) to the vertices of a graph such that no two adjacent vertices share the same color. The chromatic number is related to the graph degree: Brooks' Theorem states that for a connected graph $G$ with maximum degree $\Delta$, $\chi(G) \leq \Delta$ except when $G$ is a complete graph or an odd cycle, in which case $\chi(G) = \Delta + 1$~\cite{Brooks:1941aa,Lovasz:1975aa}. Moreover, if the maximum degree of the subgraph induced on a vertex $v\in V$ and its neighborhood is $\Delta(v)$, then $\chi(G) \geq \Delta(v) + 1$, which provides a lower bound using neighborhood degree~\cite{Diestel:book}. Lower bounds also exist in terms of the adjacency matrix~\cite{Wocjan:13,Ando:2015aa}. Finding the chromatic number is one of Karp's $21$ NP-complete problems~\cite{Karp:21-problems}, but in many cases of interest, we do not need to exactly know $\chi(G)$: since $\Delta$ is often a constant in the quantum computing context, we may in such cases simply replace $\chi(G)$ by the graph degree $\Delta$ to obtain an upper bound.

\subsection{Error model}
Consider a system of qubits occupying the vertices $V$ of a graph $G$.
We are concerned with the suppression of undesired interactions between this system and its environment and the suppression of undesired internal system terms arising, e.g., due to crosstalk. Letting $\{\sigma_v^\alpha\}_{\alpha\in\{x,y,z\}}$ denote the Pauli matrices acting only on qubit $v$, quite generally this scenario can be represented in terms of the following Hamiltonian:
\bes
\label{eq:error-hamiltonian}
\begin{align}
\label{eq:error-hamiltonian-a}
    H &= H_1 + H_2, \\
\label{eq:error-hamiltonian-b}
    H_1 &= \sum_{\alpha\in\{x,y,z\}} H_1^\alpha \ , \ \
    H_2 = \sum_{\alpha,\beta\in\{x,y,z\}} H_2^{\alpha\beta}, \\
\label{eq:error-hamiltonian-c}
    H_1^\alpha &= \sum_{v \in V} \sigma^\alpha_v \otimes B^\alpha_v \ , \ \   
    H_2^{\alpha\beta} =\!\!\!  \sum_{\{(u,v) \in E\ | \ u<v\}} \!\!\!  \sigma^\alpha_u \sigma^\beta_v \otimes B^{\alpha\beta}_{uv}.
\end{align}
\ees
In \cref{eq:error-hamiltonian-c}, we sum over ordered pairs of vertices, and $B^\alpha_v = \omega^\alpha_v \tilde{B}^\alpha_v$ ($B^{\alpha\beta}_{uv} = J^{\alpha\beta}_{uv} \tilde{B}^{\alpha\beta}_{uv}$), where $\tilde{B}^\alpha_v$ ($\tilde{B}^{\alpha\beta}_{uv}$) are dimensionless operators that act purely on the environment, and $\omega^\alpha_v$ ($J^{\alpha\beta}_{uv})$ are the corresponding couplings with dimensions of energy. The unperturbed evolution given by the ``free evolution unitary'' $f_\tau = \exp(-i\tau H)$ is then generally non-unitary and decoherent. In the case of undesired internal terms, the $\tilde{B}^\alpha_v$ ($\tilde{B}^{\alpha\beta}_{uv}$) represent the identity operator on all qubits but $v$ (and $u$). $f_\tau$ is then unitary but subject to coherent errors.

\subsection{Dynamical Decoupling}
DD is based on the application of short and narrow (``bang-bang''~\cite{Viola:98}) pulses $P_{\alpha,v}(\theta) = \exp[-i (\theta/2) \sigma_v^\alpha]$. Here we are concerned primarily with the case of $\pi$-pulses and denote $X_v = P_{x,v}(\pi) = -i\sigma_v^x$ (an ``$X$ pulse''), and similarly for $Y$ and $Z$.
It is well known that to suppress decoherence and coherent errors, it suffices to apply $\pi$ pulses at regular intervals, as long as these pulses cycle over the elements of a group~\cite{Zanardi:1999fk}. 
Applying such a sequence synchronously to all qubits would suppress $H_1$ but not $H_2$ since it would commute with the pulses.
To dynamically decouple several qubits in a manner that suppresses both $H_1$ and $H_2$, one can use Hadamard matrices and orthogonal arrays to schedule the pulses in such a way that prevents the pulses from commuting with $H_2$~\cite{Leung:01,Stollsteimer:01,Rotteler:2006aa,Wocjan:2006aa,qec-chap15,Bookatz:2016aa}.
However, such schemes are suboptimal if they do not account for the qubit connectivity graph; recent work that does resulted in a scheme requiring
a cost of $2^{\chi(G)}$ pulses to decouple all $ZZ$ crosstalk in a graph $G$~\cite{evert2024syncopated} (the special case $H=H_2^{zz}$). We now show how this can be both exponentially reduced in cost and extended to deal with a general Hamiltonian $H$ [\cref{eq:error-hamiltonian}].

\section{Results}
\label{sec:results}
We first specialize to the case where $H_1$ and $H_2$ do not include pure-$x$-type terms. This includes $ZZ$ crosstalk. 
\begin{mytheorem}[single-axis CHaDD]
\label{thm:chromatic-hadamard}
Assume the free evolution unitary is $f_\tau = e^{-i \tau H}$, where $H$ is given by \cref{eq:error-hamiltonian}. Then a circuit depth of $\chi \le N \leq 2 \chi$, involving only $X$ pulses, suffices to cancel to first order in $\tau$ 
all terms in $H$ excluding $H_1^x $ and $H_2^{xx}$
on a qubit connectivity graph $G$ with chromatic number $\chi$.
\end{mytheorem}

\begin{proof}
Let $\chi$ be the chromatic number of the graph $G$ corresponding to the Hamiltonian $H$ given by \cref{eq:error-hamiltonian}, and let $f:V\rightarrow\left\{1,2,\dots,\chi(G)\right\}\equiv C$ be a proper $\chi$-coloring of the graph $G$.
Define $V_c \equiv \{ v \in V\ |\ f(v) = c \}$, i.e, the set of all vertices of the same color $c$, and $E_{c_1,c_2} \equiv \{ (u,v)\in E\ |\ u<v,  f(u)=c_1, f(v)=c_2 \}$, i.e., the set of all ordered pairs of vertices of colors $c_1$ and $c_2$, respectively, joined by edges in $E$.
Now we return to \cref{eq:error-hamiltonian-c} and group the terms according to $G$'s coloring:
\begin{align}
\label{eq:H-regroup}
   H^\alpha_1 &=   \sum_{c \in C} \sum_{v \in V_c}  \sigma^\alpha_v \ , \ \ \ 
   H^{\alpha\beta}_2 = \sum_{c_1 \neq c_2} \sum_{(u,v) \in E_{c_1,c_2}}  \sigma^\alpha_u \sigma^\beta_v .
\end{align}
We have suppressed the bath operators in \cref{eq:H-regroup} for notational simplicity; they will not matter in our calculations below since we only consider first-order time-dependent perturbation theory through the Baker-Campbell-Hausdorff (BCH) expansion; the effect of non-commuting bath operators appears only to second order.

If we conjugate the free evolution unitary $f_\tau$ by $X$ pulses applied to all qubits of the same color $c$, i.e., by $\widetilde{X}_c \equiv \bigotimes_{v \in V_c} X_v$, we flip the sign of all $\sigma^\alpha_v$ terms corresponding to qubits $v$ of color $c$ since they anticommute with $\widetilde{X}_c$. Similarly, for the two body terms, we flip the sign of all $\sigma^\alpha_u \sigma^\beta_v$ terms where one of the qubits $u,v$ is of color $c$ and the corresponding Pauli operator anticommutes with $X$. Thus, 
\bes
\label{eq:conj-by-Xc}
\begin{align}
\label{eq:conj-by-Xc-a}
& \widetilde{X}_c f_\tau \widetilde{X}_c^\dagger = \exp \left[ -i \tau \left( \widetilde{H}_{1,c} + \widetilde{H}_{2,c} \right) \right] \\
\label{eq:conj-by-Xc-b}
& \widetilde{H}_{1,c} = \sum_{\alpha\in\{x,y,z\}} \widetilde{H}^{\alpha}_{1,c} \ , \quad
    \widetilde{H}_{2,c} = \sum_{\alpha,\beta\in\{x,y,z\}} \widetilde{H}^{\alpha\beta}_{2,c} \\
&
    \widetilde{H}^{\alpha}_{1,c} = \widetilde{X}_c H^\alpha_1 \widetilde{X}_c^\dagger =
    \sum_{c' \in C} (-1)^{\delta_{cc'}(1-\delta_{\alpha x})} \sum_{v \in V_{c'}}  \sigma^{\alpha}_v \\
    &
    \widetilde{H}^{\alpha\beta}_{2,c} = \widetilde{X}_c H^{\alpha\beta}_2 \widetilde{X}_c^\dagger = \\
&\quad \sum_{c_1 \neq c_2}
    (-1)^{\delta_{cc_1}(1-\delta_{\alpha x})}
    (-1)^{\delta_{cc_2}(1-\delta_{\beta x})}
    \!\!\!\! \sum_{(u,v) \in E_{c_1,c_2}} \!\!\!\! \sigma^\alpha_u \sigma^\beta_v \notag .
\end{align}
\ees
Note that $\widetilde{H}^x_{1,c}=H^x_1$ and $\widetilde{H}^{xx}_{2,c}=H^{xx}_2$.
To go beyond qubits of a single color $c$, 
we form a criterion for determining whether to conjugate $f_\tau$ by $\widetilde{X}_c$ for a given color $c$ at time step $j$ according to the Hadamard matrix.

Let $\nu = \left\lfloor \log_2 \chi \right\rfloor + 1$, $N=2^\nu$, and consider the $N\times N$ Hadamard matrix
\begin{align}
    W_\nu \equiv W_1^{\otimes \nu} = \frac1{\sqrt{N}} \sum_{i,j\in\{0,1\}^\nu} (-1)^{i \cdot j} \ketb{i}{j},
\end{align}
where $W_1$ is the standard $2\times 2$ Hadamard matrix and where the dot product is the bit-wise scalar product with addition modulo 2.
Below, we use $i$ and $j$ to denote an integer or its binary expansion (i.e., $i=i_0 i_i\dots i_{\nu-1}$), depending on the context.

Consider any injective function $g: C \rightarrow \{1,2,\dots,N-1\}$ that maps distinct colors in $C$ to distinct rows $i>0$ of the Hadamard matrix $W_\nu$. Note that $\nu$ is the number of bits needed to account for every color $c \in C$.
For every such color, we use row $i=g(c)$ to schedule the decoupling scheme for the qubits in $V_c$:
at each time step $j=0,1,\dots,N-1$, if $\bra{i}W_\nu\ket{j}=(-1)^{i \cdot j}=-1$, i.e., if $i \cdot j \equiv 1 \mod 2$, we conjugate $f_\tau$ by $\widetilde{X}_c$.
Equivalently, for each color $c$ at each time step $j$, we conjugate $f_\tau$ by $\widetilde{X}_c^{g(c) \cdot j}$.
The resulting unitary control propagator across all colors $c \in C$ at time step $j$ is then 
\begin{align}
\label{eq:U^x_j}
   U^x_j \equiv \overline{X}_j f_\tau \overline{X}_j^\dagger\ , \quad  \overline{X}_j \equiv \bigotimes_{c \in C} \widetilde{X}_c^{g(c) \cdot j} \ .
\end{align}
The only difference from \cref{eq:conj-by-Xc} is that now the indicator functions are generalized from a single color (such as $\delta_{cc'}$) to the set of all colors dictated by the Hadamard matrix at time step $j$, i.e., to $g(c) \cdot j$. Consequently, \cref{eq:conj-by-Xc} is replaced by
\bes
\label{eq:conj-by-Xj}
\begin{align}
\label{eq:conj-by-Xj-a}
    &
    U^x_j  =
    \exp \left[ -i \tau \left( \overline{H}_{1,j} + \overline{H}_{2,j} \right) \right] 
    \\
\label{eq:conj-by-Xj-b}
    &
    \overline{H}_{1,j}
    =
    \sum_{\alpha\in\{x,y,z\}} \overline{H}^\alpha_{1,j} \ , \ \ \ 
    \overline{H}_{2,j}
    =
    \sum_{\alpha,\beta\in\{x,y,z\}} \overline{H}^{\alpha\beta}_{2,j} 
    \\
\label{eq:conj-by-Xj-c}
    &
    \overline{H}^\alpha_{1,j}
    =
    \overline{X}_j H^\alpha_1 \overline{X}_j^\dagger =
    \sum_{c \in C} (-1)^{g(c) \cdot j(1-\delta_{\alpha x})}
    \sum_{v \in V_c}  \sigma^\alpha_v\\ 
    & 
    \overline{H}^{\alpha\beta}_{2,j}
    =
    \overline{X}_j H^{\alpha\beta}_2 \overline{X}_j^\dagger =  \\
    &\quad \sum_{c_1 \neq c_2} (-1)^{g(c_1) \cdot j(1-\delta_{\alpha x})} (-1)^{g(c_2) \cdot j(1-\delta_{\beta x})}
    \!\!\!\!\sum_{(u,v) \in E_{c_1,c_2}}  \!\!\!\! \sigma^\alpha_u \sigma^\beta_v . \notag
\end{align}
\ees
Similarly, $\overline{H}^x_{1,j}=H^x_1$ and $\overline{H}^{xx}_{2,j}=H^{xx}_2$, meaning that pure-$x$ terms are invariant.

Each $U^x_j$ takes one time step of length $\tau$. 
Then, using the BCH expansion, the entire sequence of $N$ time steps and total duration $T=N\tau$ has the following unitary:
\begin{align}
    \prod_{j=0}^{N-1} U^x_j &=
    \exp \Big[ -\! i \tau \sum_{j=0}^{N-1} \overline{H}_{1,j} + \overline{H}_{2,j}\Big]
    + \mathcal{O}\left(T^2\right) .
    \end{align}
Now, observe that the single-qubit sum 
$\sum_ j \overline{H}_{1,j} = N H^x_1$,
since $ \sum_{j=0}^{N-1} (-1)^{g(c) \cdot j}=0$ $\forall~g(c)\neq0$, where we used the fact that all rows $i=g(c)>0$ of the Hadamard matrix have entries that add up to zero. Similarly, the two-qubit sum
$\sum_ j \overline{H}_{2,j}  = N H^{xx}_2$, since $\sum_{j=0}^{N-1} (-1)^{g(c_1) \cdot j} (-1)^{g(c_2) \cdot j}=\delta_{g(c_1),g(c_2)}=0$, where we used the fact that the Hadamard matrix $W_\nu$ is orthogonal, i.e., $   \delta_{i_1,i_2}= \bra{i_1} W_\nu^T W_\nu \ket{i_2} =
\frac1{N} \sum_{j\in\{0,1\}^\nu} (-1)^{i_1 \cdot j} (-1)^{i_2 \cdot j}$, and that $g$ is injective, i.e., $c_1 \neq c_2 \implies i_1 = g(c_1) \neq g(c_2) = i_2$. Thus, 
\begin{align}
    \label{eq:single-axis-chadd-unitary}
U^x\equiv \prod_{j=0}^{N-1} U^x_j = \exp \left[ -i T ( H^x_1 + H^{xx}_2 )\right]
    + \mathcal{O}\left(T^2\right) ,
\end{align}
and we have eliminated to first order in $\tau$ all non-pure-$x$ terms in $H_1$ and $H_2$.
The total number of time steps, i.e., circuit depth, is $N=2^\nu=2^{\left\lfloor \log_2 \chi \right\rfloor + 1}$, i.e., $\chi \le N \leq 2 \chi$. 
\end{proof}

\begin{mydefinition}
The sequence $U^x = \prod_{j=0}^{N-1} U^x_j$ is called single-axis $x$-type CHaDD. Other single-axis CHaDD sequences are obtained by replacing $x$ with another axis.
\end{mydefinition}

As is the case in the single-qubit XY4 sequence~\cite{Maudsley:1986ty}, we can obtain \emph{multi-axis} CHaDD sequences by concatenating single-axis CHaDD sequences about perpendicular axes~\cite{Khodjasteh:2005xu}. This leads a concatenated multi-axis CHaDD protocol requiring a quadratic circuit depth in $\chi(G)$; see \cref{app:A}. However, as we show next, for $\chi(G)\ge 3$, a more efficient multi-axis CHaDD protocol is possible that retains the linear scaling in $\chi(G)$. For $\chi(G)=2$, both protocols require a circuit depth of $16$.

Rather than associating a color with a single row of a Hadamard matrix, we need three rows per color for multi-axis CHaDD. To explain the protocol we need to first introduce some additional terminology~\cite{Leung:01,Stollsteimer:01,Rotteler:2006aa,Wocjan:2006aa,qec-chap15,Bookatz:2016aa}.
A sign matrix $S_{\chi,N}$ 
is a $\chi\times N$ matrix with $\pm 1$ entries ($W_\chi$ is a special case). The Schur product $C = A\circ B$ of two $\chi\times N$ sign matrices $A$ and $B$ is the entry-wise product $C_{ij} = A_{ij}B_{ij}$. A set of $\chi\times N$ sign matrices is Schur-closed if it is closed under the Schur product. A \emph{Schur subset} is a set of three different rows of $S_{\chi,N}$ that multiply entry-wise to $++\cdots +$. For example, $\{(+--),(-+-),(--+)\}$ is a Schur subset of some $S_{\chi\ge 3,3}$. Note that Schur subsets are Schur-closed. We will identify Schur subsets with colors.

\begin{mytheorem}[multi-axis CHaDD]
\label{thm:chromatic-sign}
Under the same assumptions as in \cref{thm:chromatic-hadamard}, a circuit depth of $3\chi+1 \le N \le 2(3\chi+5)$, involving only single-qubit Pauli pulses, suffices to cancel to first order in $\tau$ \emph{all} terms in $H$ on a qubit connectivity graph $G$ with chromatic number $\chi$.
\end{mytheorem}

The proof combines ideas from Refs.~\cite{Leung:01,qec-chap15} with our single-axis CHaDD approach. 

\begin{proof}
Assign to each color $c$ a Schur subset via an injective function $h$. For a given Schur subset $h(c)$, label its three rows $h(c;x)$, $h(c;y)$, and $h(c;z)$. Define the $\chi \times N$ sign matrix $S_\alpha$, where $\alpha\in\{x,y,z\}$, as the matrix whose rows are $\{h(c;\alpha)\}_{c \in C}$. These sign matrices $S_\alpha$ are Schur-closed. This implies that for each fixed $(c,j)$'th entry, $([S_x]_{c j},[S_y]_{c j},[S_z]_{c j})$ can only be one of $(+,+,+)$, $(+,-,-)$, $(-,+,-)$, or $(-,-,+)$.
These tuples correspond to whether or not a given Pauli operator $\widetilde\sigma_{(c,j)}$ acting on all qubits of color $c$ commutes (+) or anticommutes (-) with $(\widetilde{X}_c,\widetilde{Y}_c,\widetilde{Z}_c)$, i.e., $\widetilde\sigma_{(c,j)}=\widetilde{I}_c$, $\widetilde{X}_c$, $\widetilde{Y}_c$, and $\widetilde{Z}_c$, respectively.
The signs acquired after conjugating Pauli terms $(\sigma^\alpha_v)_{\alpha\in\{x,y,z\},f(v)=c}$ by $\widetilde\sigma_{(c,j)}$ are the $(c,j)$'th entries $([S_\alpha]_{c j})_{\alpha\in\{x,y,z\}}$.

Then the unitary control propagator is $\overline\sigma_j \equiv \bigotimes_{c \in C} \widetilde\sigma_{(c,j)}$, and when we conjugate $f_\tau$ by $\overline\sigma_j$, i.e., $U_j \equiv \overline\sigma_j f_\tau \overline\sigma_j^\dagger$,
the sign acquired by $\sigma_v^\alpha$ is $[S_\alpha]_{c j}$ if $f(v)=c$,
while the sign acquired by $\sigma_u^\alpha \sigma_v^\beta$ ($u<v$) is $[S_\alpha]_{c_1 j} [S_\beta]_{c_2 j}$ if $f(u)=c_1$ and $f(v)=c_2$:
\bes
\label{eq:conj-by-sigmaj}
\begin{align}
\label{eq:conj-by-sigmaj-a}
    &
    U_j  =
    \exp \left[ -i \tau \left( \overline{H}_{1,j} + \overline{H}_{2,j} \right) \right] 
    \\
\label{eq:conj-by-sigmaj-b}
    &
    \overline{H}_{1,j}
    =
    \sum_{\alpha\in\{x,y,z\}} \overline{H}^\alpha_{1,j} \ , \ \ \ 
    \overline{H}_{2,j}
    =
    \sum_{\alpha,\beta\in\{x,y,z\}} \overline{H}^{\alpha\beta}_{2,j} 
    \\
\label{eq:conj-by-sigmaj-c}
    &
    \overline{H}^\alpha_{1,j}
    =
    \overline\sigma_j H^\alpha_1 \overline\sigma_j^\dagger =
    \sum_{c \in C} [S_\alpha]_{c j}
    \sum_{v \in V_c}  \sigma^\alpha_v\\ 
    & 
    \overline{H}^{\alpha\beta}_{2,j}
    =
    \overline\sigma_j H^{\alpha\beta}_2 \overline\sigma_j^\dagger =  \\
    &\quad \sum_{c_1 \neq c_2} [S_\alpha]_{c_1 j} [S_\beta]_{c_2 j}
    \!\!\!\!
    \sum_{(u,v) \in E_{c_1,c_2}}  \!\!\!\! \sigma^\alpha_u \sigma^\beta_v . \notag
\end{align}
\ees
The overall unitary for all $N$ time steps is now
\begin{align}
    U \equiv \prod_{j=0}^{N-1} U_j &=
    \exp \Big[ -\! i \tau \sum_{j=0}^{N-1} \overline{H}_{1,j} + \overline{H}_{2,j}\Big]
    + \mathcal{O}\left(T^2\right) .
\end{align}
To remove all one- and two-local terms, i.e., $\sum_j \overline{H}_{1,j} = 0$ and $\sum_j \overline{H}_{2,j} = 0$,  we need $\sum_j [S_\alpha]_{c j} = 0$ and $\sum_j [S_\alpha]_{c_1 j} [S_\beta]_{c_2 j} = 0$. The first condition follows immediately from the fact that these are rows of the Hadamard matrix, and the second follows from the fact that each Schur-subset was drawn from different partitions ($c_1 \ne c_2$) of the Hadamard matrix, so the rows of the sign matrices $S_\alpha$ are orthogonal. 

The Hadamard matrix $W_\nu$ can be partitioned into exactly $S_e = (2^\nu-1)/3$ Schur-subsets if $\nu$ is even and at most $S_o = (2^\nu-5)/3$ if $\nu$ is odd~\cite{Leung:01} (see also \cite[Thm.~4.1]{Beutelspacher:1975aa} and \cref{app:B}). Let $\nu=2k$ and $k\in\mathbb{Z}^+$, so $S_e(k) = (2^{2k}-1)/3$ and $S_o = (2^{2k+1}-5)/3$. Since the circuit depth is $N=2^\nu$, we would like to minimize $k$. We need at least as many Schur subsets as colors, hence $\min\{S_e(k),S_o(k)\}\ge \chi$. Minimizing $k$, we have $k = \min\{f_e(\chi), f_o(\chi) \}$, where $f_e(\chi) = \lceil \frac{1}{2} \log_2(3\chi+1) \rceil $ and $f_o(\chi) = \lceil \frac{1}{2} [\log_2(3\chi+5) -1]\rceil$. Thus, we obtain $N = \min\{2^{2\lceil \frac{1}{2} \log_2(3\chi+1) \rceil} , 2^{2\lceil \frac{1}{2} [\log_2(3\chi+5) -1]\rceil+1} \} \in \Theta(\chi)$, i.e., 
$3\chi+1 \le N \le 2(3\chi+5)$.
\end{proof}

We explain the upper and lower bounds in \cref{app:C}. Since $N$ is a step function in $\chi$, these bounds are loose (e.g., $N=16$ for $2\le \chi \le 5$, and $N=32$ for $6\le \chi \le 9$), but they are nearly tight envelopes. 

\subsection{Examples}
We illustrate single-axis $x$-type CHaDD for a few examples with low chromatic numbers.
Since $\chi(G)\geq2$ for graphs with $|E|\geq1$, we use graphs $G$ with $\chi(G)=2,3$, i.e., $\nu = \left\lfloor \log_2 \chi(G) \right\rfloor + 1 = 2$. 
By \cref{thm:chromatic-hadamard}, a circuit depth of $N=2^\nu=4$ is then sufficient to cancel all non-pure-$x$ interactions (e.g., $ZZ$ crosstalk) and qubit decoherence to first order, scheduled according to the Hadamard matrix $W_2$ shown in \cref{fig:g-W2} (middle).

\begin{figure}[t]
\centering
\includegraphics[width=0.98\columnwidth]{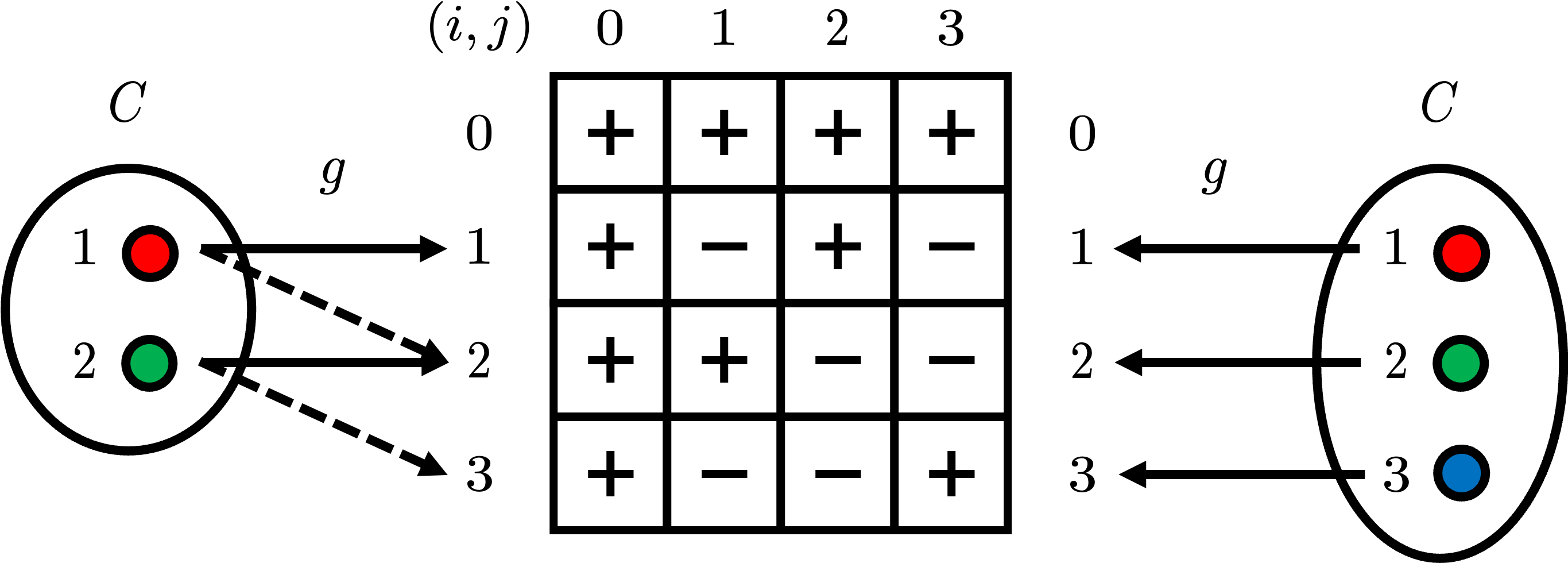}
\caption{Color-to-row maps for $2$-colorable (left) and $3$-colorable (right) graphs. Solid arrows correspond to an unbalanced schedule, dashed to a balanced schedule. Middle: the Hadamard matrix $W_2$.
}
\label{fig:g-W2}
\end{figure}

Per \cref{thm:chromatic-hadamard}, the first row ($i=0$) is not used. 
For $\chi(G)=3$, we must use all three of the remaining rows $i=1,2,3$, but for $\chi(G)=2$, we can choose any pair of rows.

\subsubsection{Example 1: $W_2$ for $\chi(G)=2$}
Consider a square grid of qubits with nearest-neighbor coupling, as in \cref{fig:schedules}  (top and middle).
This graph $G$ is $2$-colorable (bipartite), i.e., there exists a proper $2$-coloring $f:V\rightarrow\{1,2\} \equiv C$.
Examples of such graphs include the Rigetti Ankaa family of chips~\cite{qcs-qpus}
and the IBM heavy-hex layout~\cite{ibm-heavy-hex}. 

Let us choose the identity color-to-row map $i=g(c)=c$ so that colors $c=1,2$ correspond to rows $i=1,2$, respectively, as in \cref{fig:g-W2} (left, solid arrows).
We conjugate the free evolution $f_\tau$ by all $\widetilde{X}_i$ operators for which $\left<i\right|W_2\left|j\right>=-1$ for each time step (column) $j$. This readily yields 
$U^x_0=f_\tau$ for the first interval,
$U^x_1 = \widetilde{X}_1 f_\tau \widetilde{X}_1^\dagger$ for the second,
$U^x_2=\widetilde{X}_2 f_\tau \widetilde{X}_2^\dagger$ for the third, and, since
$\overline{X}_3=\widetilde{X}_1\widetilde{X}_2$, 
finally $U^x_3 = (\widetilde{X}_1\widetilde{X}_2)f_\tau(\widetilde{X}_1\widetilde{X}_2)^\dagger$ for the last interval.
Taking the product of all unitaries yields the \emph{unbalanced} schedule
\begin{align}
    U^x = \prod_{j=0}^3 U^x_j
    &=
    \widetilde{X}_1\widetilde{X}_2 f_\tau \widetilde{X}_1 f_\tau \widetilde{X}_1 \widetilde{X}_2 f_\tau \widetilde{X}_1 f_\tau ,
\label{eq:chi2-unbalanced}
\end{align}
wherein a pulse is applied to all $1$-colored qubits every $\tau$ but only every $2\tau$ to the $2$-colored qubits; see \cref{fig:schedules} (top).

A \emph{balanced} schedule results if instead we choose $i=g(c)=c+1$, as in \cref{fig:g-W2} (left, dashed arrows),
resulting in $U^x_0=f_\tau$,
$U^x_1=\widetilde{X}_2 f_\tau \widetilde{X}_2^\dagger$, 
$U^x_2=(\widetilde{X}_1\widetilde{X}_2) f_\tau (\widetilde{X}_1\widetilde{X}_2)^\dagger$, 
and $U^x_3=\widetilde{X}_1 f_\tau \widetilde{X}_1^\dagger$, so that
\begin{align}
    U^x = \prod_{j=0}^3 U^x_j
= \widetilde{X}_1 f_\tau \widetilde{X}_2 f_\tau \widetilde{X}_1 f_\tau \widetilde{X}_2 f_\tau ,
\label{eq:chi2-balanced}
\end{align}
wherein a pulse is applied alternately to $1$ and $2$-colored qubits every $2\tau$. This schedule is illustrated in \cref{fig:schedules} (middle) and may be considered preferable due to its symmetry and lower overall pulse count of four single-color pulses \textit{vs} six for the unbalanced schedule. It is an interesting question to identify the general conditions for a balanced schedule directly from the Hadamard matrix in arbitrary dimensions. 

\begin{figure}[t]
        {\includegraphics[width=0.23\columnwidth]{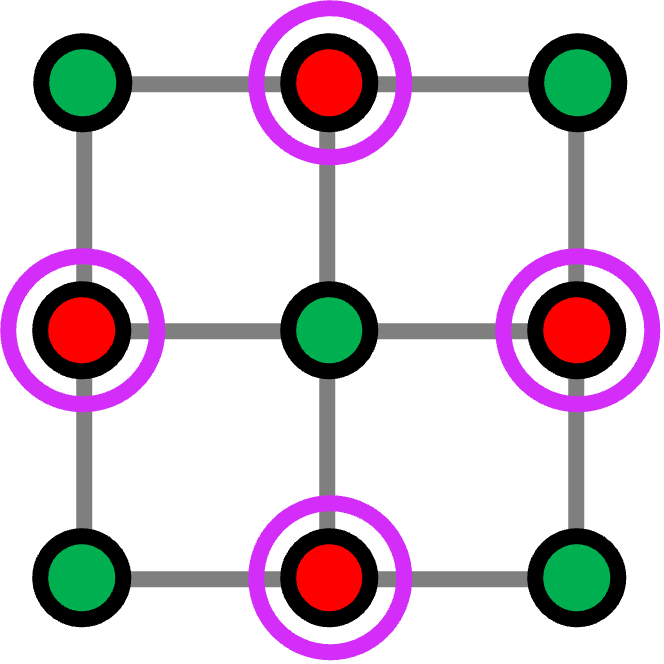}
        \label{fig:chi2-unbalanced-schedule-a}}
        {\includegraphics[width=0.23\columnwidth]{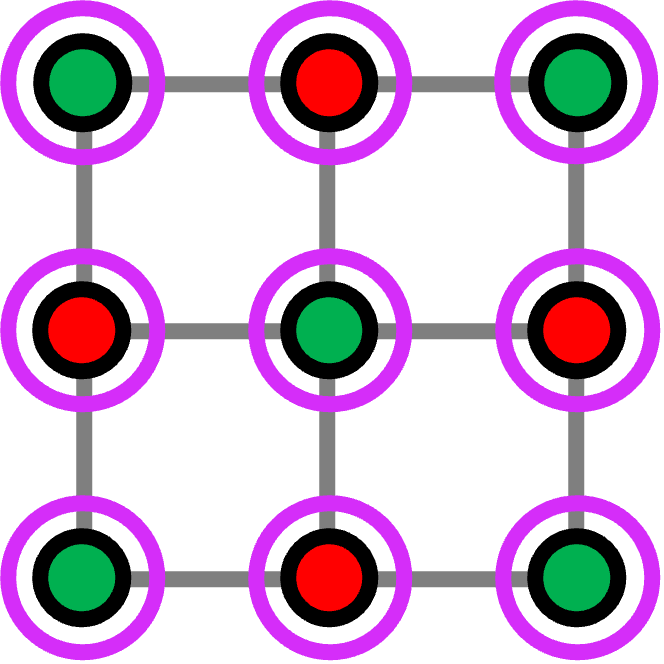}
        \label{fig:chi2-unbalanced-schedule-b}}
        {\includegraphics[width=0.23\columnwidth]{chi2-conj1}
        \label{fig:chi2-unbalanced-schedule-c}}
        {\includegraphics[width=0.23\columnwidth]{chi2-conj12}
        \label{fig:chi2-unbalanced-schedule-d}}\\
        \vspace{.15cm}
        {\includegraphics[width=0.23\columnwidth]{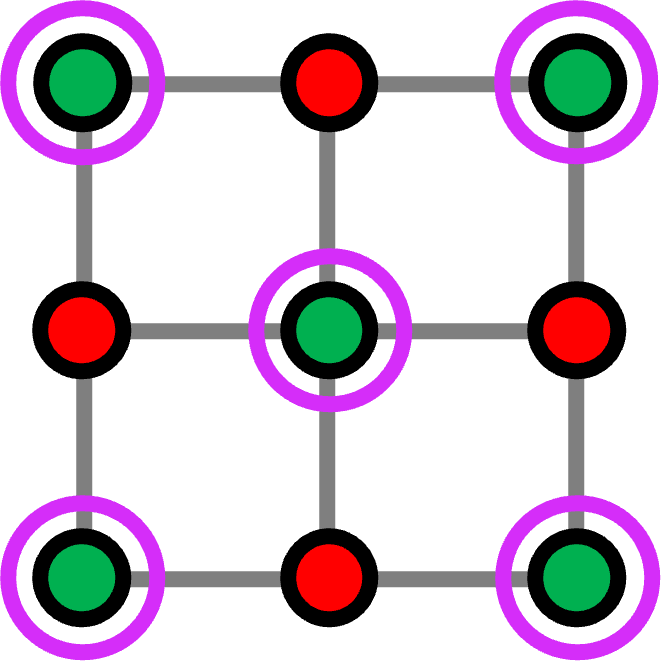}
        \label{fig:chi2-balanced-schedule-a}}
        {\includegraphics[width=0.23\columnwidth]{chi2-conj1}
        \label{fig:chi2-balanced-schedule-b}}
        {\includegraphics[width=0.23\columnwidth]{chi2-conj2}
        \label{fig:chi2-balanced-schedule-c}}
        {\includegraphics[width=0.23\columnwidth]{chi2-conj1}
        \label{fig:chi2-balanced-schedule-d}}\\
        \vspace{.15cm}
         {\includegraphics[width=0.23\columnwidth]{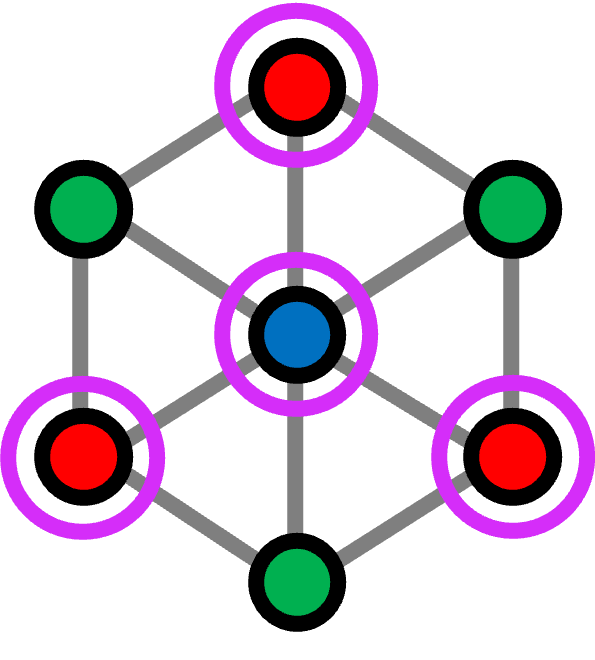}
        \label{fig:chi3-schedule-a}}
        {\includegraphics[width=0.23\columnwidth]{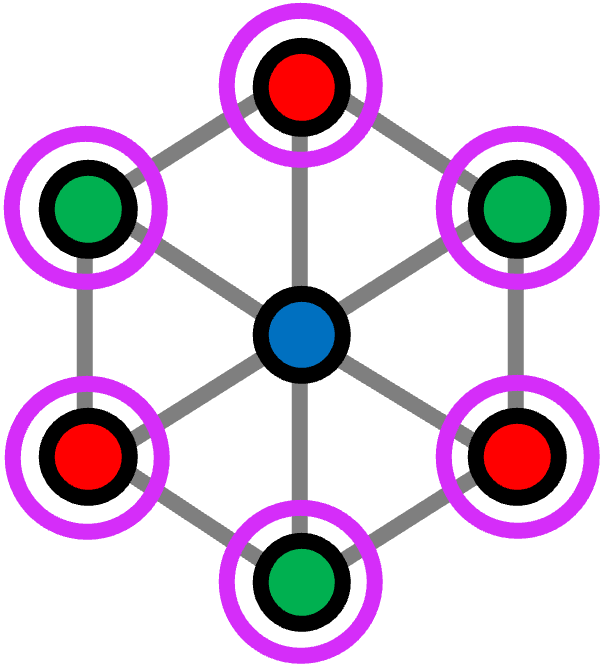}
        \label{fig:chi3-schedule-b}}
        {\includegraphics[width=0.23\columnwidth]{chi3-conj13}
        \label{fig:chi3-schedule-c}}
        {\includegraphics[width=0.23\columnwidth]{chi3-conj12}
        \label{fig:chi3-schedule-d}}
    \caption{Examples of schedules for 2-colorable (top and middle) and 3-colorable (bottom) graphs. 
    Purple circles indicate DD pulses, and time flows from left to right. Top row: unbalanced schedule for a $2$-colorable graph [\cref{eq:chi2-unbalanced}]. Middle row: balanced schedule for a $2$-colorable graph [\cref{eq:chi2-balanced}]. Bottom row: unbalanced schedule for a $3$-colorable graph [\cref{eq:chi3-schedule}].}
    \label{fig:schedules}
\end{figure}

\subsubsection{Example 2: $W_2$ for $\chi(G)=3$}
An example for which $\chi(G)=3$ is a triangular grid [\cref{fig:schedules} (bottom)].
We must now use all three rows $i=1,2,3$ of the Hadamard matrix $W_2$, and $g$ can only be a permutation. Consider the identity: $i=g(c)=c$, as in \cref{fig:g-W2} (right). We obtain $U^x_0=f_\tau$, $\overline{X}_1=\widetilde{X}_1\widetilde{X}_3$ so that $U^x_1=(\widetilde{X}_1\widetilde{X}_3) f_\tau (\widetilde{X}_1\widetilde{X}_3)^\dagger$, and similarly $U^x_2=(\widetilde{X}_2\widetilde{X}_3) f_\tau (\widetilde{X}_2\widetilde{X}_3)^\dagger$ and $U^x_3=(\widetilde{X}_1\widetilde{X}_2) f_\tau (\widetilde{X}_1\widetilde{X}_2)^\dagger$. Thus,
\begin{align}
    U^x = \prod_{j=0}^3 U^x_j = 
\widetilde{X}_1 \widetilde{X}_2 f_\tau \widetilde{X}_1 \widetilde{X}_3 f_\tau \widetilde{X}_1 \widetilde{X}_2 f_\tau \widetilde{X}_1 \widetilde{X}_3 f_\tau ,
\label{eq:chi3-schedule}
\end{align}
an unbalanced schedule using eight single-color pulses; see \cref{fig:schedules} (bottom).

\begin{figure}[b]
% \centering
(a)
\includegraphics[width=0.44\columnwidth]{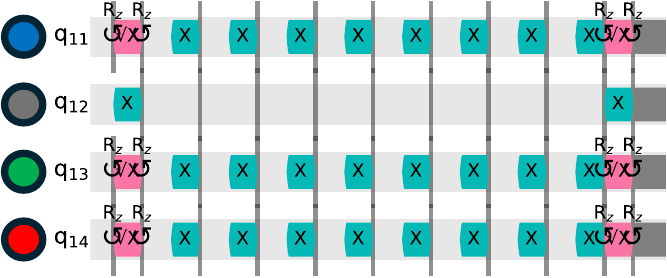}
(b)
\includegraphics[width=0.44\columnwidth]{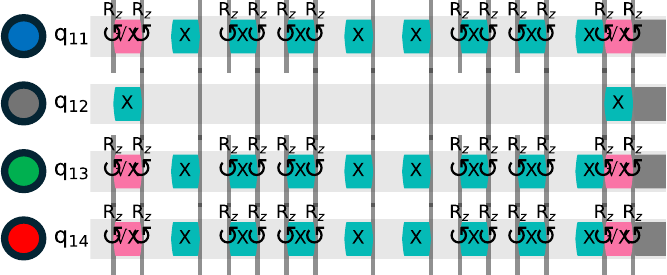} \\
(c)
\includegraphics[width=0.44\columnwidth]{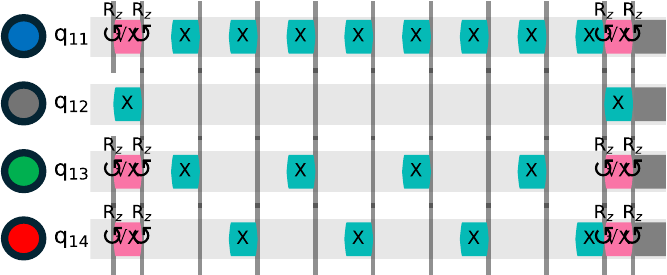}
(d)
\includegraphics[width=0.44\columnwidth]{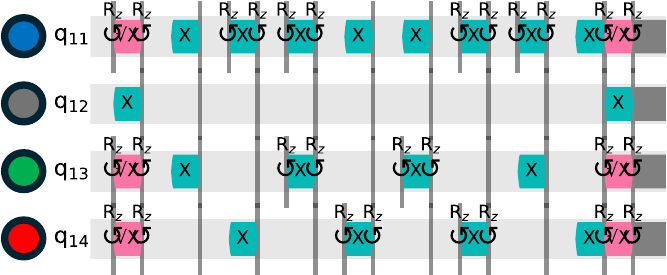} \\
(e)
\includegraphics[width=0.65\columnwidth]{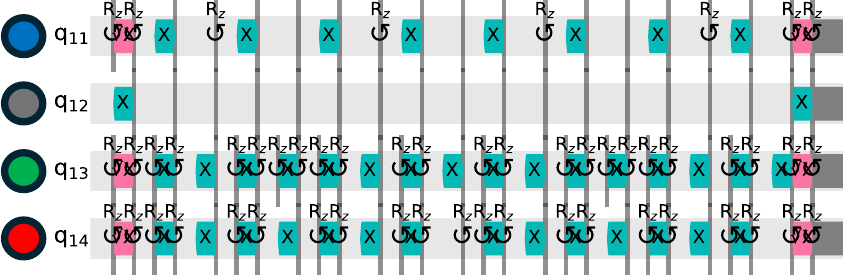}
\caption{
A 4-qubit portion of the 127-qubit $\chi=3$ experiment timeline of (a) XX, (b) UR4, (c) CHaDD, (d) CHaDD-R, and (e) multi-axis CHaDD with red, green, blue (RGB), and gray qubits prepared in the $\ket{+}$ and $\ket{1}$ states, respectively.
XX, UR4, and CHaDD are length-4 sequences repeated twice, CHaDD-R is a length-8 sequence repeated once, and multi-axis CHaDD is a length-16 sequence repeated once. The CHaDD sequences (c) and (d) are scheduled according to the color-to-row function $g(1)=2, g(2)=3, g(3)=1$, different from that in the bottom row of \cref{fig:schedules}, and use two pulses per time-slot compared to three for their achromatic counterparts (a) and (b), i.e., use $1/3$ fewer pulses, resulting in a PRR that is $2/3$ that of XX and UR4. During (e) multi-axis CHaDD, the RGB qubits see an average of $38/3$ pulses over 16 time steps, resulting in a PRR that is $38/48 \approx .79$ that of XY4.
To avoid interference effects arising from placing pulses too close together, we set $\tau=2\delta$ (where $\delta$ is the pulse width), and to ensure a correct implementation of $\bar{X}$, we use a symmetric decomposition ($\bar{X} = R_z(-\pi) X R_z(\pi)$, where $R_z(\varphi) \equiv \exp(-i\varphi \sigma^z/2)$ is a virtual-$Z$ gate~\cite{McKay:2017tv}) for UR4 and CHaDD-R~\cite{vezvaee2024virtualzgatessymmetric}.
}
\label{fig:bivalent-3coloring-timing}
\end{figure}

\subsection{A robust CHaDD variant}
DD sequences with robustness to pulse imperfections are known~\cite{Quiroz:2013fv,Genov:2017aa}. The ``universally robust'' (UR$n$) family of sequences, for even $n\geq4$, consists of $n$ $\pi$-pulses about various axes in the $xy$-plane~\cite{Genov:2017aa}. The simplest sequence, UR4, consists of the uniform interval sequence $X \bar{X} \bar{X} X$, where $X = P_x(\pi)$ and $\bar{X} = P_x(-\pi)$. Analogously, we define a robust version of the single-axis $x$-type CHaDD sequence by replacing $X$ with $\bar{X}$ where necessary to match the UR4 sequence pattern.

For example, the $\chi=3$ CHaDD sequence given in \cref{eq:chi3-schedule} consists of colors 1, 2, and 3 receiving $XXXX$, $IXIX$, and $XIXI$, respectively.
We need a minimum of four pulses on each color to form a robust sequence, so we replace these with $X \bar{X} \bar{X} X X \bar{X} \bar{X} X$, $IXI \bar{X} I \bar{X} I X$, and $XI \bar{X} I \bar{X} IXI$, which doubles the circuit depth compared to CHaDD. We denote this robust CHaDD variant ``CHaDD-R4,'' or CHaDD-R for brevity. Note that CHaDD-R promotes the underlying achromatic UR4 sequence to an efficient multi-qubit sequence in the same sense that CHaDD and multi-axis CHaDD respectively promote the XX and XY4 sequences.

\textit{Pulse repetition rate at fixed chromatic number}.---%
\cref{thm:chromatic-hadamard} and \cref{thm:chromatic-sign} characterize CHaDD in terms of the scaling of circuit depth with chromatic number. However, with an eye towards an experimental comparison between the CHaDD variants and their underlying achromatic sequences at \emph{fixed} chromatic number, it is also meaningful to quantify performance in terms of pulse repetition rate (PRR), i.e., the average number of pulses per unit time.
A lower PRR is advantageous, as it means reduced power consumption and, hence, less heat being injected into the system from microwave sources. A lower PRR also implies a reduction in coherent error accumulation. The PRR for $\chi=3$ is illustrated for a few key sequences in \cref{fig:bivalent-3coloring-timing}.
Since the achromatic sequences (XX, XY4, UR4, etc.) do not distinguish qubits by color, they target all qubits in every time-step, while CHaDD variants selectively target qubits by color. The result is a PRR that is $1/3$ lower for CHaDD and CHaDD-R than it is for XX and UR4, respectively. 
Clearly, a sequence achieving equal performance at lower PRR is preferred. We next describe experiments designed to test the performance of CHaDD variants relative to achromatic sequences according to this criterion.

\section{Experiments}
\label{sec:experiments}

We performed experiments on the ibm\_brisbane 127 qubit QPU. 
IBM's ``heavy-hex'' qubit-connectivity graphs are 2-colorable, as they lack an odd cycle. Previous work on crosstalk cancellation~\cite{evert2024syncopated, Zhou2023} did not present explicit schemes for $\chi> 2$ and instead focused on 2-colorable graphs (including Rigetti QPUs). To demonstrate CHaDD's ability to go beyond $\chi=2$, we embed within the 2-colorable heavy-hex lattice a 3-colorable system graph. This is illustrated in \cref{fig:bivalent-3coloring}, with ``system'' qubits colored red, green, and blue (RGB). The remaining ``bath'' qubits are colored gray. 

\begin{figure}[t]
\centering
\includegraphics[width=0.98\columnwidth]{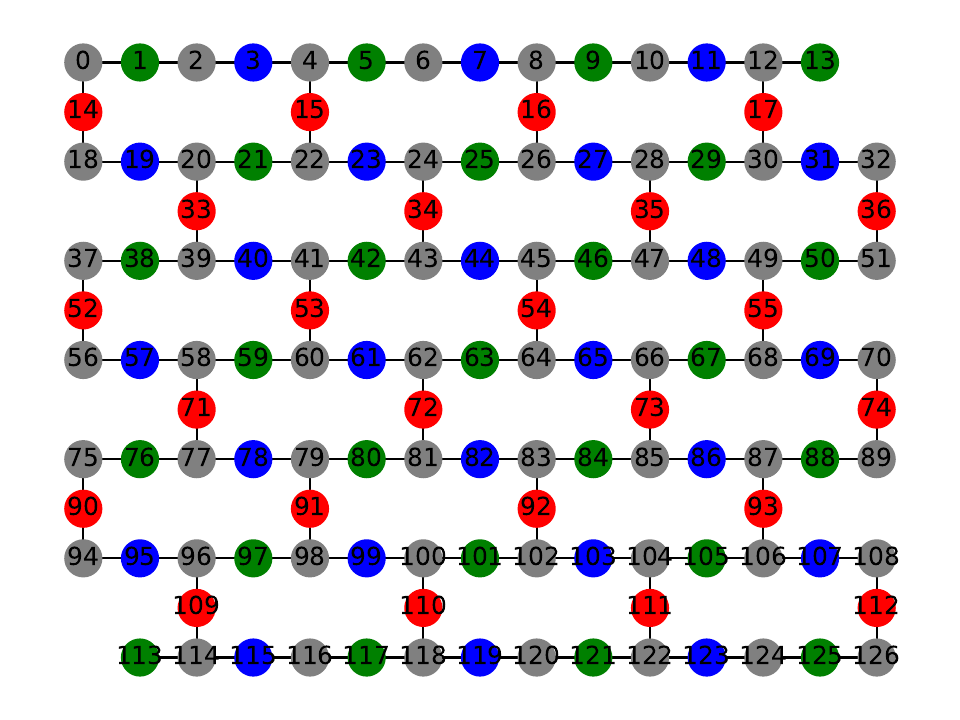}
\caption{Embedded 3-coloring of the 127 qubit heavy-hex graph consisting of system qubits (red, green, blue) and bath qubits (gray).
}
\label{fig:bivalent-3coloring}
\end{figure}

What distinguishes the system from the bath qubits is that our goal is to suppress both decoherence and crosstalk on the former but not the latter. With this in mind, we recall from the proof of 
\cref{thm:chromatic-hadamard} that the first row of the Hadamard matrix ($i=0$) is not used to schedule any pulses. 
Thus, we assign the gray qubits to this row and the red, green, and blue qubits to rows 2, 3, and 1 of $W_2$, respectively.
In this manner, decoherence of the gray qubits is not suppressed, though their crosstalk still is (since row 0 is orthogonal to all the other rows).
In addition, while the gray qubits are measured, we discard the corresponding measurement outcomes. 

In order to compare different DD sequences, we initialize all qubits in $\ket{0}$, prepare the gray qubits in each of $\ket{0}$, $\ket{+}$, and $\ket{1}$, and prepare the RGB qubits in each of the $6$ Pauli eigenstates, for a total of $18$ different preparation state settings. We perform DD on the RGB qubits, undo the state preparation, and obtain the RGB qubits' fidelity as the fraction of $\ket{0}$ states we obtain when measuring all RGB qubits in the $\sigma^z$ basis over 5,000 trials, or ``shots.'' XX, XY4, UR4, and CHaDD are length-$4$ sequences, CHaDD-R is a length-$8$ sequence, and multi-axis CHaDD is a length-$16$ sequence.
Accordingly, we repeat XX, XY4, UR4, and CHaDD in multiples of $4$ up to $80$ times, CHaDD-R in multiples of $2$ up to $40$ times, and multi-axis CHaDD up to $20$ times so that each sequence results in $21$ data points that take the same amount of time to collect. We report the average over the $18$ preparation state settings. 

\begin{figure}[t]
\centering
\includegraphics[width=0.98\columnwidth]{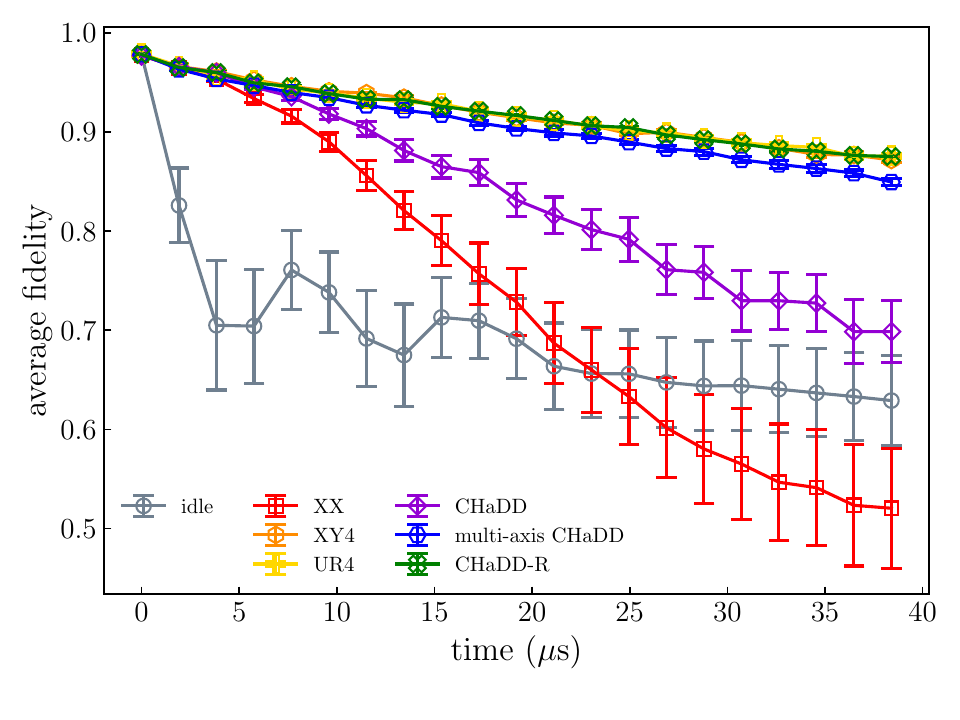}
\caption{
Comparison of average initial state preservation fidelity of RGB qubits, averaged over 18 preparation state settings, for all achromatic DD sequences and their CHaDD counterparts run on ibm\_brisbane.
Error bars here and in all other plots are 1$\sigma$.
}
\label{fig:1215-all}
\end{figure}

\begin{figure}[t]
\centering
\includegraphics[width=0.98\columnwidth]{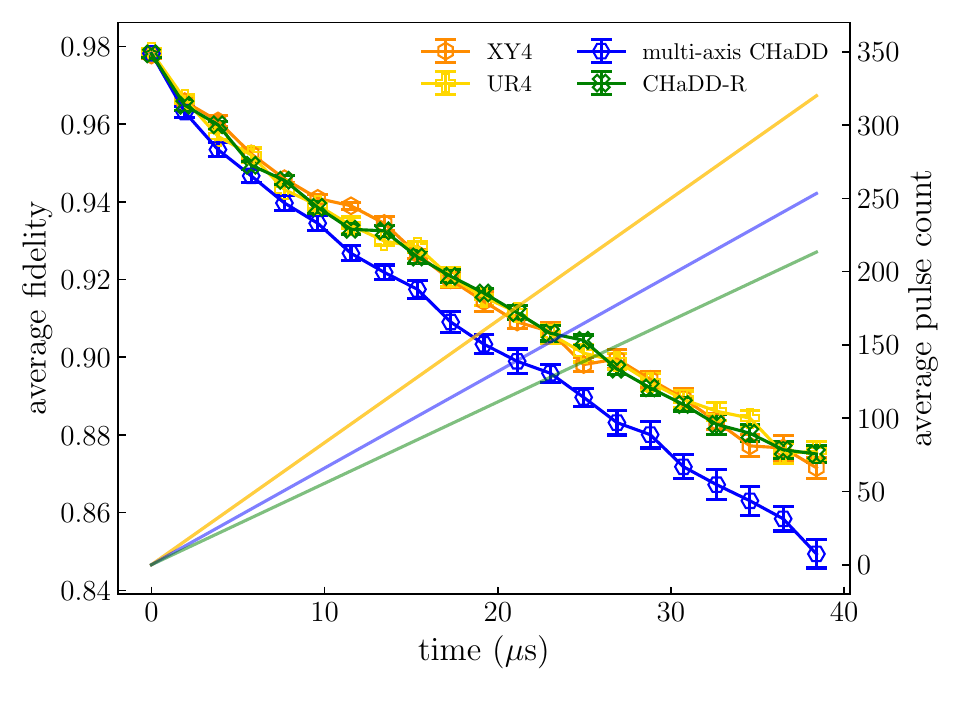}
\caption{
Average state preservation fidelity for RGB qubits (left) and pulse count per RGB qubit (right) of the most performant DD sequences on ibm\_brisbane.
Multi-axis CHaDD and CHaDD-R are the CHaDD versions of XY4 and UR4, respectively, and they exhibit similar performance, but CHaDD-R has $1/3$ fewer pulses than UR4 and XY4, while multi-axis CHaDD has $1/5$ fewer pulses than XY4  
and UR4 (see \cref{fig:bivalent-3coloring-timing}). Note that the yellow and orange lines overlap since XY4 and UR4 have the same PRR.
}
\label{fig:1215-best}
\end{figure}

\cref{fig:1215-all} collects the results from three achromatic DD sequences and three CHaDD variants, along with free evolution (idle). XY4, UR4, multi-axis CHaDD, and CHaDD-R are the top-performing sequences. Note that these sequences have much smaller $1\sigma$ error bars than the others, indicating less susceptibility to variation in initial state conditions and a better ability to preserve arbitrary initial states. 
We compare the top sequences in \cref{fig:1215-best}, where XY4, UR4, and CHaDD-R are seen to be roughly tied for top performance.

CHaDD, multi-axis CHaDD, and CHaDD-R can be thought of as the CHaDD variants of the underlying achromatic XX, XY4, and UR4 sequences, respectively. In this sense, the CHaDD variants either outperform or are on par with their underlying achromatic sequences, most notably CHaDD \textit{vs} XX in \cref{fig:1215-all}. Crucially, the CHaDD variants achieve this performance with a significantly reduced pulse count (or PRR), as also shown in \cref{fig:1215-best}. As noted above, CHaDD-R has a $1/3$ smaller PRR than UR4. As anticipated, this constitutes a much sparser and hence more power-efficient sequence for the same performance. A lower PRR also reduces coherent error accumulation, which is the likely explanation for the large performance gap between CHaDD and XX, given the latter's absence of inherent robustness against coherent errors (in contrast, XY4 and UR4 are both robust sequences~\cite{Genov:2017aa}). At the same time, CHaDD lacks a coherent error suppression mechanism, which explains why it underperforms the robust sequences.

While all the chromatic and achromatic sequences decouple all the heavy-hex RGB-gray couplings, only the chromatic sequences additionally suppress the red-green, red-blue, and green-blue next-nearest neighbor couplings; they do this with a lower PRR than the achromatic sequences, as illustrated in \cref{fig:bivalent-3coloring-timing}. The fact that the achromatic sequences also decouple the RGB qubits is a peculiarity of the embedding; crosstalk on a native 3-colorable graph would not be decoupled by synchronous, achromatic sequences, whereas with the embedding the gray qubits mediate the crosstalk between the RGB qubits, and an achromatic sequence that is applied only to the RGB qubits, as we do here, does suppress crosstalk as a result. 
To see this, consider, e.g., the interaction $\sigma^z_R  \sigma^z_{\text{gray}} + \sigma^z_{\text{gray}}  \sigma^z_B$, which anticommutes with, and hence is suppressed by, the XX-sequence pulse $(X_R \otimes I_{\text{gray}}) (I_{\text{gray}}\otimes X_B)$ applied only the red (R) and blue (B) qubits. Thus, the embedding creates an artificially favorable setup for the achromatic sequences.
To demonstrate this, we repeated our experiments while letting the latter also target the gray qubits. The result is shown in \cref{fig:RGB-vs-not}, and it confirms that the performance of the achromatic sequences is significantly reduced, with strong crosstalk oscillations clearly visible.

\cref{app:D} provides additional experimental results obtained using other IBM QPUs that are in complete qualitative agreement with the conclusions reported in this section.

\begin{figure}[t]
\centering
\includegraphics[width=0.98\columnwidth]{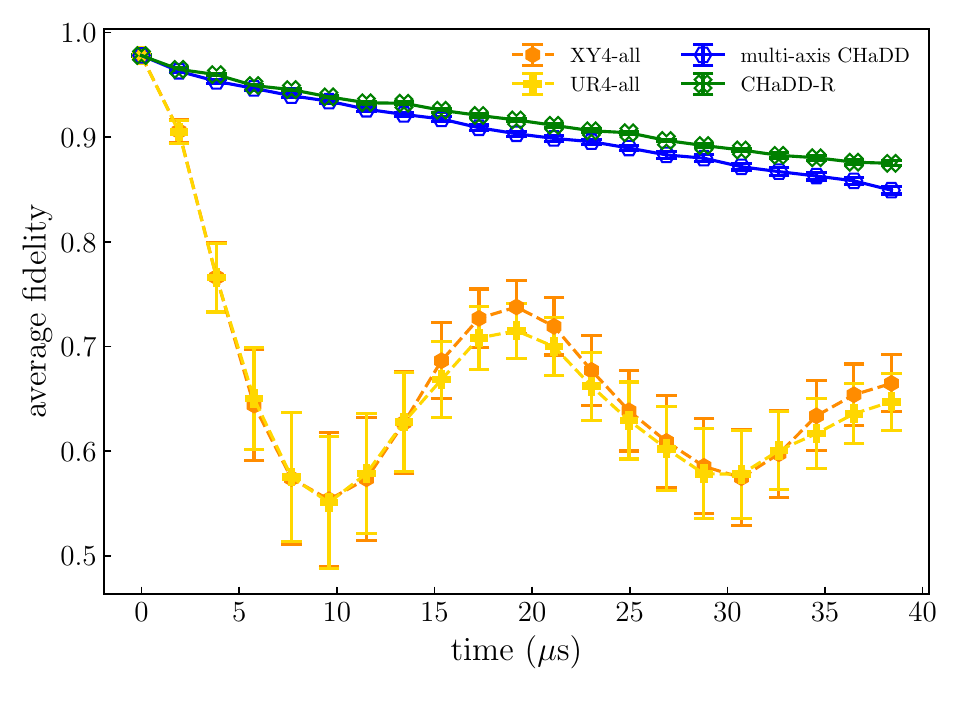}
\caption{As in \cref{fig:1215-all} but with the achromatic sequences XY4 and UR4 applied to all qubits. The multi-axis CHaDD and CHaDD-R data is the same as in \cref{fig:1215-all}. Both CHaDD variants significantly outperform their underlying achromatic sequences in this setting.
}
\label{fig:RGB-vs-not}
\end{figure}

\section{Conclusions and outlook}
\label{sec:conc}
We have proposed CHaDD along with both its multi-axis and robust variants. These are first-order DD schemes that are efficient in the chromatic number $\chi$ of the qubit connectivity graph, which, as a function of $\chi$, lead to a scaling advantage in the circuit depth required to decouple arbitrary connectivity graphs relative to previously known multi-qubit decoupling methods~\cite{Leung:01,Stollsteimer:01,Rotteler:2006aa,Wocjan:2006aa,qec-chap15,Bookatz:2016aa,Paz-Silva:2016aa,Zhou2023, niu2024multiqubit, evert2024syncopated}. This includes, as a special case, multi-qubit $ZZ$ crosstalk decoupling. At fixed $\chi$, we have demonstrated a significant reduction in the pulse repetition rate -- and hence, power consumption -- required to achieve a performance that is on par with or better than standard achromatic DD sequences.

To generalize CHaDD beyond first-order decoupling is straightforward. For example, concatenation of multi-axis CHaDD with itself is the direct multi-qubit generalization of concatenating XY4 with itself, which leads to the single-qubit concatenated DD (CDD) family~\cite{Khodjasteh:2005xu}, yielding an extra suppression order in time-dependent perturbation theory with every additional level of concatenation~\cite{Khodjasteh:2007zr,Ng:2011dn}. 

However, concatenation incurs an exponential cost in circuit depth, and it would be desirable to be able to replace the uniform pulse intervals used in the construction of single-axis CHaDD with the non-uniform intervals used in the single-qubit Uhrig DD (UDD) sequence family~\cite{Uhrig:2007qf}, in order to guarantee the same higher-order performance as UDD~\cite{UL:10} in the single-axis, multi-qubit setting. This would directly extend to multi-axis single-qubit decoupling via the quadratic DD (QDD) sequence~\cite{West:2010:130501}, which requires at most $(m+1)^2$ pulses for order-$m$ suppression in terms of the Dyson series expansion~\cite{Xia:2011uq}. Successfully combining CHaDD with QDD would generally be much more efficient than the multi-qubit, nested UDD (NUDD) sequence, which scales exponentially with the number of qubits~\cite{Wang:10}. Whether these extensions of CHaDD are possible is an interesting open problem.

It should be possible to generalize CHaDD beyond qubits to multi-level systems by replacing the Pauli matrices used in the proof of \cref{thm:chromatic-hadamard} with elements of the generalized Pauli (or Heisenberg-Weyl) group~\cite{Nielsen:2002ab,Wang:2020ab,Wocjan:2006aa}. CHaDD-R, our robust version of CHaDD, is based on the UR4 sequence; a proper generalization of UR$n$~\cite{Genov:2017aa} to CHaDD-R$n$ is an interesting topic for future research. Finally, it is another interesting open problem to generalize CHaDD from two-local and instantaneous pulses to $k$-local interactions and bounded controls for $n$ qubits, for which the current state of the art is $\mathcal{O}[n^{k-1}\log(n)]$~\cite{Bookatz:2016aa}.

\acknowledgments
This material is based upon work supported by, or in part by, the U.S. Army Research Laboratory and the U.S. Army Research Office under contract/grant number W911NF2310255.
This research was supported by the Office of the Director of National Intelligence (ODNI), Intelligence Advanced Research Projects Activity (IARPA) and the Army Research Office, under the Entangled Logical Qubits program through Cooperative Agreement Number W911NF-23-2-0216.
The authors gratefully acknowledge Victor Kasatkin and Jenia Mozgunov for helpful discussions.

\appendix

\section{Multi-axis CHaDD via concatenation}
\label{app:A}

\begin{mycorollary}[Concatenated multi-axis CHaDD]
\label{multi-axis-chromatic-hadamard}
Assume the free evolution unitary is $f_\tau = e^{-i \tau H}$, where $H$ is given by \cref{eq:error-hamiltonian}. Then a circuit depth of $4^{\left\lfloor \log_2 \chi(G) \right\rfloor + 1} \leq 4 \chi^2(G)$ is sufficient to completely cancel $H$ to first order in $\tau$ on a qubit connectivity graph $G$ with chromatic number $\chi(G)$.
\end{mycorollary}

\begin{proof}
It follows immediately from \cref{thm:chromatic-hadamard} (by swapping the $x$ and $z$ indices) that a single-axis $z$-type CHaDD sequence has an overall unitary of
\begin{align}
    U^z \equiv \prod_{j=0}^{N-1} U^z_j = \exp \left[ -i T ( H^z_1 + H^{zz}_2 )\right] + \mathcal{O}\left(T^2\right) .
\end{align}
If we concatenate the single-axis $x$- and $z$-type sequences at every time step, i.e., replace $f_\tau$ in $U^x_j = \overline{X}_j f_\tau \overline{X}_j^\dagger$ with the single-axis $z$-type CHaDD sequence $U^z$, we obtain $U^{xz}_j \equiv \overline{X}_j  U^z \overline{X}_j^\dagger$, so that 
\begin{align}
    \label{eq:multi-axis-chadd-unitary-unsimplified}
    U^{xz} \equiv \prod_{j=0}^{N-1} U^{xz}_j = \prod_{j=0}^{N-1} \overline{X}_j e^{-i T (H^z_1 + H^{zz}_2)} \overline{X}_j^\dagger
    + \mathcal{O}\left(N T^2\right) .
\end{align}
This is just the single-axis $x$-type CHaDD sequence $U^x$ given by \cref{eq:single-axis-chadd-unitary} applied to an effective Hamiltonian $H = H^z_1 + H^{zz}_2$ without any $H_1^x$ or $H_2^{xx}$ terms, so it follows from \cref{thm:chromatic-hadamard} that this Hamiltonian is completely canceled to first order, i.e.,
\begin{align}
\prod_{j=0}^{N-1} \overline{X}_j e^{-i T ( H^z_1 + H^{zz}_2 )} \overline{X}_j^\dagger = I + \mathcal{O}\left(N T^2\right) .
\label{eq:multi-axis-chadd-unitary}
\end{align}
Hence, the total CHaDD unitary is $U^{xz} = I + \mathcal{O}\left(N T^2\right)$.
This requires a circuit depth of $N^2 = 2^{2(\left\lfloor \log_2 \chi(G) \right\rfloor + 1)} \leq 4 \chi^2(G)$.
\end{proof}

\begin{mydefinition}
The sequence $U^{xz} = \prod_{j=0}^{N-1} U^{xz}_j$ is called concatenated multi-axis CHaDD.
\end{mydefinition}

Clearly, replacing $x$ and $z$ with other combinations of orthogonal axes is equivalent to $U^{xz}$. Thus, $U^{xz}$ can be interpreted as the efficient multi-axis, multi-qubit generalization of the XY4 sequence~\cite{Maudsley:1986ty}.

\section{Connections to projective geometry}
\label{app:B}

The properties of Schur subsets of a sign matrix presented by Leung~\cite{Leung:01} were discovered in a more general form several decades earlier in the projective geometry community, using the terminology of partial $t$-spreads in the Desargesian projective space $PG(d, q)$, where $d$ is the dimension and $q$ is the order (number of elements) of the finite field $\mathbb{F}_q$~\cite{Beutelspacher:1975aa}. In such a space, points and lines are defined in terms of vector spaces over $\mathbb{F}_q$, and the incidence structure obeys the axioms of projective geometry. Desargues' theorem guarantees that any configuration of points and lines that forms a perspective triangle has collinear points, ensuring a well-defined projective structure. A perspective triangle is a configuration where two triangles are in perspective from a point or a line

For our purposes, $PG(d, q) = \mathbb{Z}_2^{\nu} \setminus \{0\}$, $d = \nu-1$, $q=2$, $t=1$, and the size of the partial $t$-spread is the number of Schur subsets.

In this context, Ref.~\cite[Thm.~4.1]{Beutelspacher:1975aa} states that there are no more than $(2^{\nu} - 5) / 3$ Schur subsets for a Hadamard matrix $W_\nu$ with odd $\nu$. For even $\nu$ one can pick all $2^{\nu} - 1$ non-trivial rows of the Hadamard matrix (one row is the trivial row of all $+$), resulting in $(2^{\nu} - 1) / 3$ Schur subsets, so the upper bound is not needed.

Moreover, the bounds for both odd and even $\nu$ are achievable: for even $\nu$ Ref.~\cite{Beutelspacher:1975aa} states that this was known before (Results 2.1 and 2.2), and for odd $\nu$ this is shown in Theorem 4.2.

\section{Circuit depth of multi-axis CHaDD}
\label{app:C}

In the main text, we showed that 
\begin{align}
N = \min\{2^{2\lceil \frac{1}{2} \log_2(3\chi+1) \rceil} , 2^{2\lceil \frac{1}{2} [\log_2(3\chi+5) -1]\rceil+1} \}
\label{eq:N-multi}
\end{align}
for the circuit depth of multi-axis CHaDD. Here we derive the upper and lower bounds.
It follows from the definitions of $S_e$ and $S_o$ that $2^\nu \geq 3\chi + 1$ and $2^\nu \geq 3\chi + 5$ in the case that $\nu$ is even and odd, respectively.
Thus, in either case, we have that $N = 2^\nu \geq 3\chi + 1$, establishing a lower bound on the circuit depth regardless of the value of $\nu$.
Now consider that we wish to choose the smallest value of $\nu$ that satisfies the above inequalities so that we may state $2^\nu \geq 3\chi + 1 \geq 2^{\nu-1}$ and $2^\nu \geq 3\chi + 5 \geq 2^{\nu-1}$ if $\nu$ is even or odd, respectively, from which it follows that $2(3\chi + 5) \geq 2^\nu = N$.
Hence, $2(3\chi + 5) \geq N \geq 3\chi + 1$.
These results are shown in \cref{fig:N-multi}. 
%While the bounds are quite loose, they are nearly tight envelopes.

\begin{figure}[t]
\centering
\includegraphics[width=0.98\columnwidth]{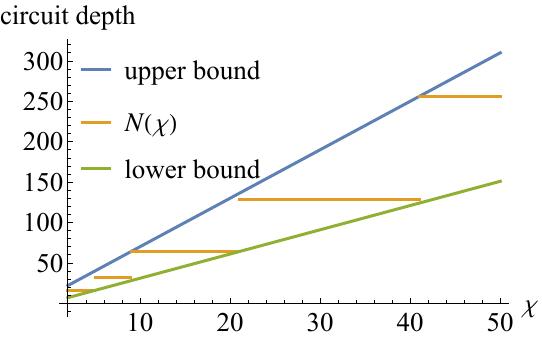}
\caption{Circuit depth $N(\chi)$ of multi-axis CHaDD [\cref{eq:N-multi}] along with the upper and lower bounds, $2(3\chi+5)$ and $3\chi+1$, respectively.}
\label{fig:N-multi}
\end{figure}

\section{Additional experimental results}
\label{app:D}

In the main text, we presented results obtained on the ibm\_brisbane QPU. In this section, to test the robustness of our conclusions, we provide results from additional IBM QPUs. These results, presented in \cref{fig:device-comparison}, are in full qualitative agreement with our conclusions. Namely, in all cases, the CHaDD variants (rightmost column in the legend) outperform their underlying achromatic sequences (middle column in the legend). 

\begin{figure*}
    \centering
    \subfigure[\ ibm\_strasbourg]{\includegraphics[width=0.44\linewidth]{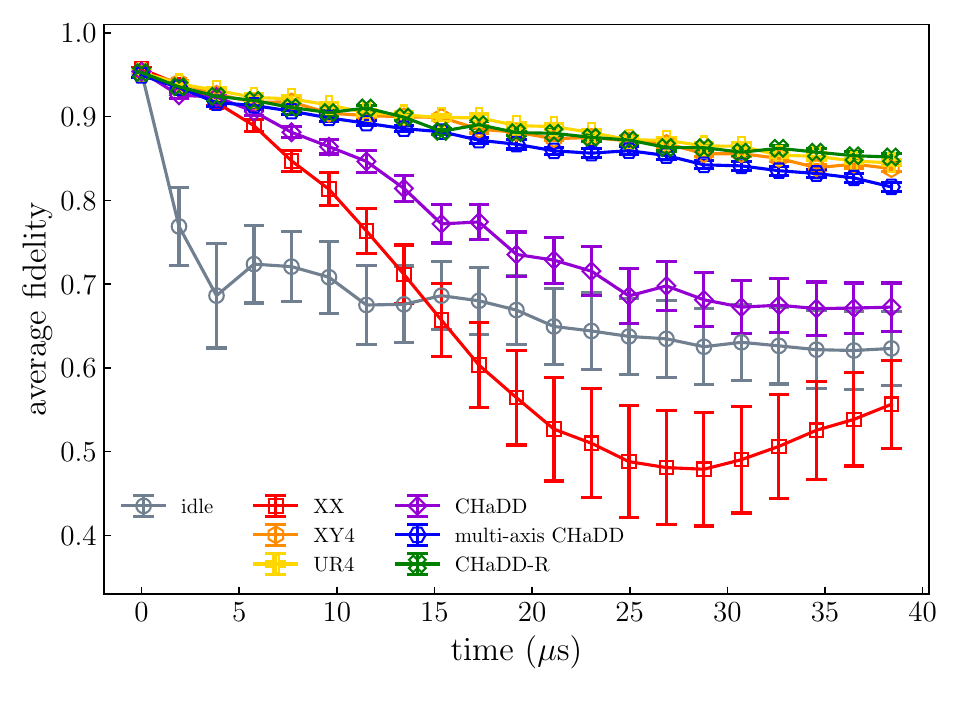}}
    \subfigure[\ ibm\_brisbane]{\includegraphics[width=0.44\linewidth]{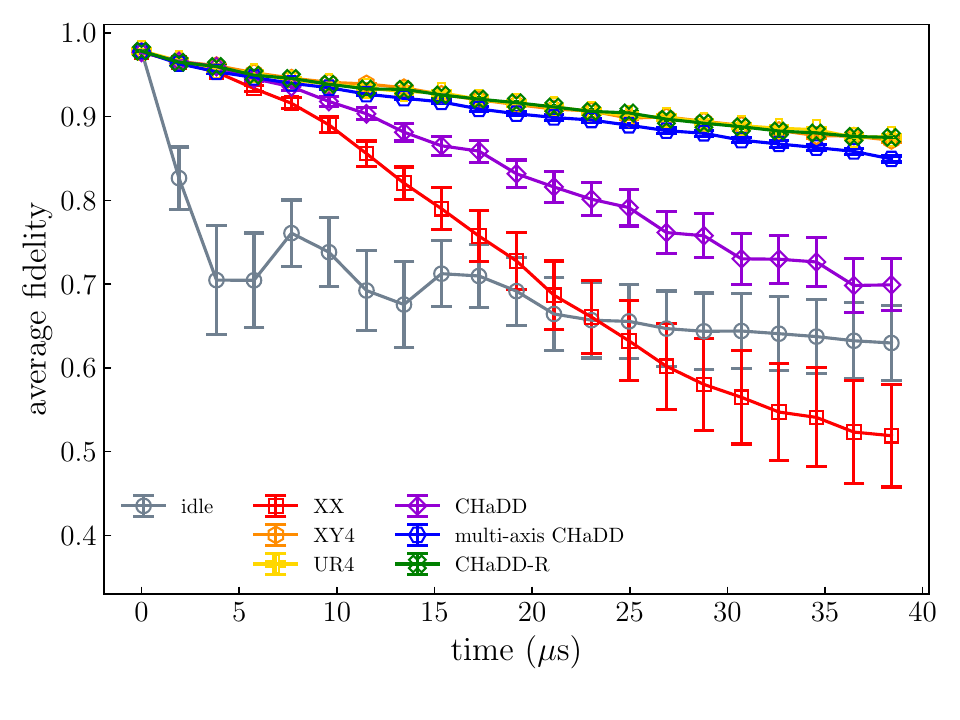}}
    \subfigure[\ ibm\_kyiv]{\includegraphics[width=0.44\linewidth]{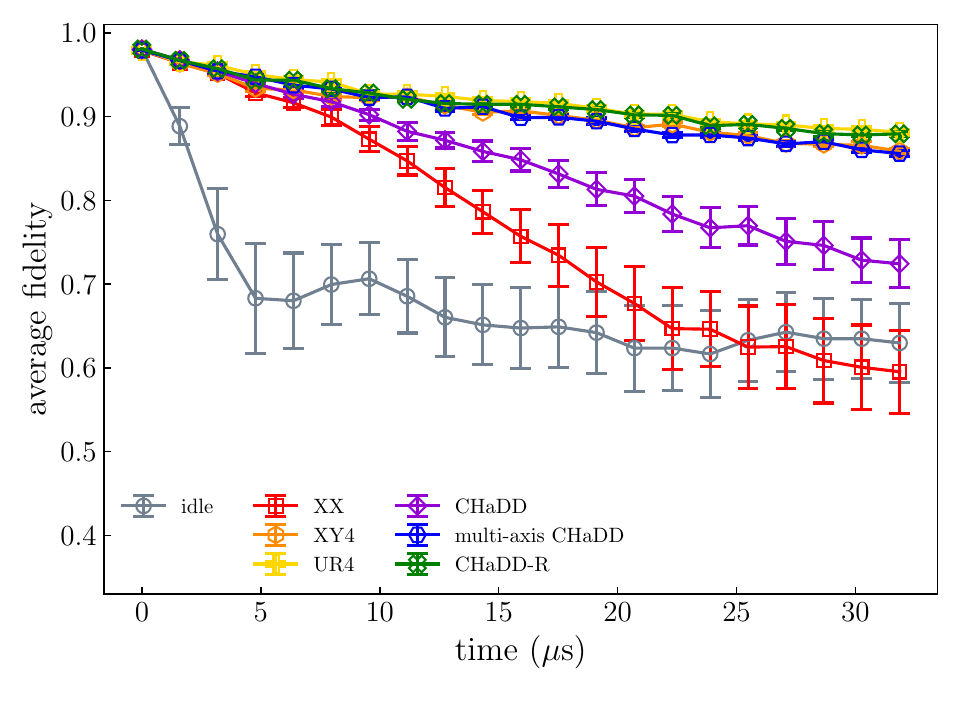}}
    \subfigure[\ ibm\_sherbrooke]{\includegraphics[width=0.44\linewidth]{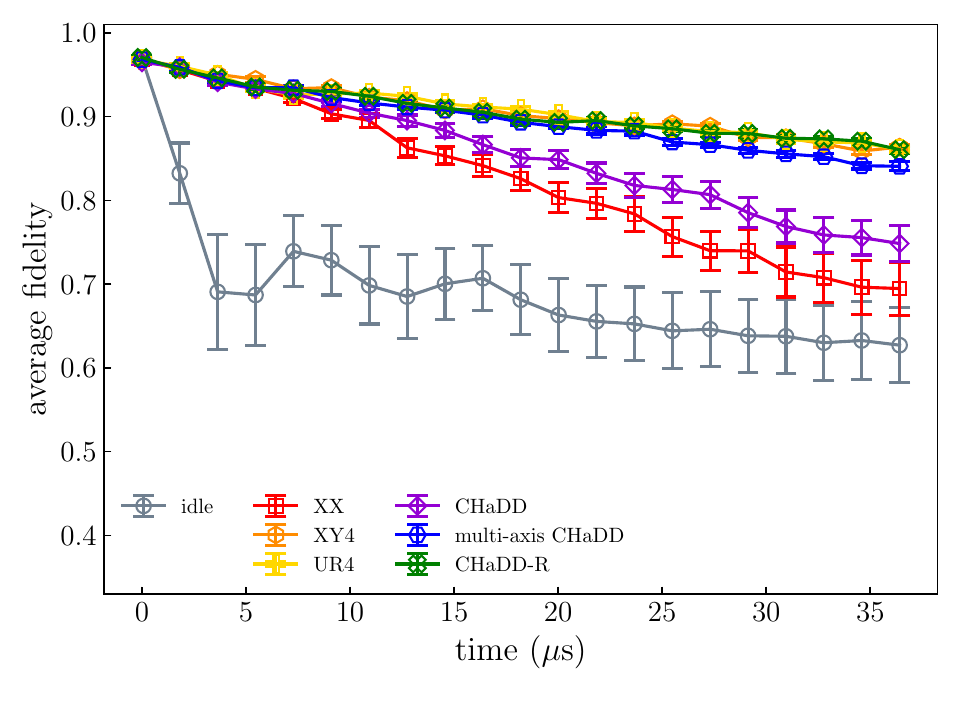}}
    \caption{Comparison across four IBM QPUs. The time axes are different since the duration of single-qubit gates varies: (a) ibm\_strasbourg: $60$ ns, (b) ibm\_brisbane: $60$ ns, (c) ibm\_kyiv: $49.78$ ns (d) ibm\_sherbrooke: $56.89$ ns. Different QPUs exhibit different quantitative behaviors, with ibm\_strasbourg displaying the overall lowest fidelities and largest crosstalk, evidenced by the oscillation in the XX sequence results.
    }
    \label{fig:device-comparison}
\end{figure*}

\clearpage
\newpage 
\bibliographystyle{apsrev4-2}

\end{document}